\documentclass[submission,copyright,creativecommons]{eptcs}
\providecommand{\event}{GANDALF 2021} 
\usepackage{breakurl}                 
\usepackage{underscore}               

\usepackage[utf8]{inputenc}
\usepackage[T1]{fontenc}
\usepackage{comment}
\usepackage{bookmark}          
\usepackage{url}               
\usepackage[all]{foreign}      
\usepackage{booktabs}          
\usepackage{cite}
\usepackage{multicol}
\usepackage{marginnote}        
\usepackage[only,llbracket,rrbracket]{stmaryrd} 
\usepackage[inline]{enumitem}
\usepackage{amsmath,amssymb, amsthm}   
\usepackage{thmtools, thm-restate}
\usepackage{packages/commands} 
\usepackage{packages/symbols}  
\usepackage[capitalise,nameinlink]{cleveref}    
\usepackage{tikz}                
\usepackage{subcaption}        
\usepackage{tcolorbox}         
\usepackage{soul}              
\usepackage{listings}          
\usepackage{lstautogobble}
\usepackage{xcolor}
\usepackage{wrapfig}           
\usepackage{tcolorbox}         
\usepackage{adjustbox}

\newtheorem{definition}{Definition}
\newtheorem{theorem}{Theorem}
\newtheorem{lemma}{Lemma}
\newtheorem*{lemma*}{Lemma}
\newtheorem{corollary}{Corollary}
\newtheorem{proposition}{Proposition}
\newtheorem{remark}{Remark}

\Crefname{corollary}{Cor.}{Cors.}
\Crefname{proposition}{Prop.}{Props.}
\Crefname{theorem}{Th.}{Ths.}
\Crefname{definition}{Def.}{Defs.}

\newcommand\blthanks[2][]{%
  \begingroup
  \renewcommand\thefootnote{}\footnote{%
    \makebox[0pt][r]{%
      \scriptsize #1\hspace{0.5em}%
    }#2%
  }%
  \addtocounter{footnote}{-1}%
  \endgroup
}

\title{Expressiveness of Extended Bounded Response \LTL}

\author{Alessandro Cimatti
\institute{Fondazione Bruno Kessler\\ Trento, Italy}
\email{cimatti@fbk.eu}
\and
Luca Geatti
\institute{Fondazione Bruno Kessler\\ Trento, Italy}
\institute{University of Udine\\ Udine, Italy}
\email{lgeatti@fbk.eu}
\and
Nicola Gigante
\institute{Free University of Bozen-Bolzano\\ Bolzano, Italy}
\email{nicola.gigante@unibz.it}
\and
Angelo Montanari
\institute{University of Udine\\ Udine, Italy}
\email{angelo.montanari@uniud.it}
\and
Stefano Tonetta
\institute{Fondazione Bruno Kessler\\ Trento, Italy}
\email{tonettas@fbk.eu}
}
\def\titlerunning{Expressiveness of Extended Bounded Response LTL}
\def\authorrunning{Cimatti, Geatti, Gigante, Montanari and Tonetta}

\begin{document}
\maketitle

\begin{abstract}
  Extended Bounded Response \LTL with Past (\LTLEBRP) is a safety fragment of
  Linear Temporal Logic with Past (\LTLP) 
  that has been recently introduced in the context of reactive synthesis.  The
  strength of \LTLEBRP is a fully symbolic compilation of formulas into
  symbolic deterministic automata.
  Its syntax is organized in four levels. The first three levels feature (a
  particular combination of) future temporal modalities, the last one admits
  only past temporal operators. At the base of such a structuring there are
  algorithmic motivations: each level corresponds to a step of the algorithm
  for the automaton construction.
  The complex syntax of \LTLEBRP made it difficult to precisely
  characterize its expressive power, 
  and to compare it with other \LTLP safety fragments.
\blthanks[]{
  The work was partially supported by the Italian GNCS project \emph{a
  Cognitive frAmework to synTHesize Adaptive contRollers in Safety-crItical
  Scenarios (CATHARSIS)} (L.~Geatti, N.~Gigante and A.~Montanari).}

  In this paper, we first prove that \LTLEBRP is expressively complete with
  respect to the safety fragment of \LTLP, that is, any safety language
  definable in \LTLP can be formalized in \LTLEBRP, and vice versa.
  From this, it follows that \LTLEBRP and \safetyltl are expressively
  equivalent.
  Then, we show that past modalities play an essential role in \LTLEBRP: we
  prove that the future fragment of \LTLEBRP is strictly less expressive than
  full \LTLEBRP.
\end{abstract}


\section{Introduction}
\label{sec:intro}

Linear Temporal Logic (\LTL) was introduced in the late seventies
\cite{pnueli1977temporal} as a modal logic for reasoning over computer
programs, modeling their computations as state sequences (\ie linear orders)
that represent the state a computer program is in at a given time.
\LTL originally used temporal modalities for moving only in the future of
a time point. Later, it turned out that adding modalities for moving in the
past (we refer to this logic as \LTLP) does not add expressive power to \LTL
\cite{lichtenstein1985glory}, but only succinctness \cite{markey2003temporal}.
The definition of the operators in the syntax of \LTLP was proved to be
carefully designed. In fact, Kamp \cite{kamp1968tense} as well as Gabbay \etal
\cite{gabbay1980temporal} proved that the properties that one can formalize in
\LTLP are exactly those definable in the first-order fragment of the
\emph{monadic second-order theory of one successor} (\SoneS, for short), which
is in turn decidable \cite{buchi1960decision,buchi1960weak}.

Among the different properties that one can define in \LTLP, two notable
classes are the set of \emph{safety} and \emph{co-safety} properties.  Safety
properties express the intuitive requirement that \emph{something bad never
happens}, and thus each counterexample of a safety property is finite.
Co-safety properties are duals of safety properties: each state sequence that
satisfies the property has a finite witness.  The safety and co-safety classes
play a crucial role in verification and synthesis, since their main feature of
having \emph{finite} witnesses makes in general the problems much simpler
\cite{kupferman2001model,ZhuTLPV17}.

Several safety fragments of \LTL have been introduced over the years. One of
the most natural examples is \safetyltl
\cite{sistla1994safety,ChangMP92,ZhuTLPV17}.  The \safetyltl logic is defined
as the set of all and only those formulas of \LTL (with only future modalities)
such that, when in negated normal form, do \emph{not} contain existential
temporal operators (like the \emph{until} operator). In \cite{ChangMP92}, Chang
\etal proved that all the safety properties definable in \LTL are expressible
in \safetyltl as well, and \viceversa.

Extended Bounded Response \LTL with Past (\LTLEBRP) is a recently introduced
safety fragment of \LTLP with an efficient reactive synthesis problem. 
In addition to the fact that realizability from \LTLEBRP specifications is
\EXPTIME-complete (while \LTL realizability is \EXPTIME[2]-complete), in
practice realizability and synthesis from \LTLEBRP specifications turned out to
be much more efficient than other approaches~\cite{cimatti2020reactive}.
The syntax of \LTLEBRP is articulated over layers: the first three layers
comprise a combination of future temporal modalities, while the last layer
includes only past temporal operators. Each of the layers was carefully
designed in order to correspond to a step of the algorithm for constructing
a symbolic automaton starting from an \LTLEBRP specification.
This results into a great performance improvement in practice, but the syntax
of \LTLEBRP makes it hard to find its exact expressive power, 
and, consequently, makes it hard also to compare it with other safety
fragments of \LTLP, like, for instance, \safetyltl.

In this paper we prove that \LTLEBRP is expressively complete with respect to
the safety fragment of \LTLP. As a by-product, we obtain that \LTLEBRP and
\safetyltl are expressively equivalent.
The core of the proof exploits a \emph{normal form theorem} for each safety
property definable in \LTLP \cite{thomas1988safety,ChangMP92}, which
establishes a correspondence between safety properties definable in \LTLP and
properties of the form $\ltl{G \alpha}$, where $\ltl{G}$ is the \emph{globally}
operator of \LTL and $\alpha$ is a pure past formula.
Consequently, it is clear that the \emph{pure past layer} of \LTLEBRP
plays a crucial role for the expressive equivalence of \LTLEBRP. 
We show that this layer is really necessary. In fact, we prove that \LTLEBR,
that is \LTLEBRP devoid of the pure past layer, is \emph{strictly less}
expressive than full \LTLEBRP.
This is shown by proving that all the formulas of \LTLEBR can constrain, for
any time point $i$ in an infinite state sequence, only a \emph{bounded} prefix
before (or interval around) $i$.
This implies that formulas that are able to constrain, for each time point $i$,
a prefix of unbounded (although finite) length before $i$, like for instance
$\ltl{G(p_1 \to H p_2)}$ (where $\ltl{H}$ is the \emph{historically} past
operator of \LTLP), are \emph{not} definable in \LTLEBR.

The rest of the paper is organized as follows.  In \cref{sec:prelim}, we give
the necessary background.  The expressive power of \LTLEBRP is proved in
\cref{sec:ltlebr}.  In \cref{sec:futureltlebr}, we prove that the future
fragment of \LTLEBRP is strictly less expressive than \LTLEBRP. Finally, we
summarize the results of the paper in \cref{sec:conc}.


\section{Preliminaries}
\label{sec:prelim}

In this section, we give the definitions that are necessary throughout the
paper.

\subsection{Linear Temporal Logic}

Linear Temporal Logic (\LTL) is a modal logic interpreted over infinite,
discrete linear orders \cite{pnueli1977temporal,demri2016temporal}.
Syntactically, \LTL can be seen as an extension of propositional logic with the
addition of  the \emph{next} operator ($\ltl{X \phi}$, \ie, at the \emph{next}
state $\phi$ holds) and the \emph{until} operator ($\ltl{\phi_1 U \phi_2}$, \ie,
$\phi_2$ will eventually hold and $\phi_1$ will hold \emph{until} then).

\emph{\LTL with Past} (\LTLP) extends \LTL with the addition of temporal
operators able to talk about what happened in the \emph{past} with respect to
the current time, and it is obtained from \LTL by adding the following
\emph{past} temporal operators:
\begin{enumerate*}[label=(\roman*)]
  \item the \emph{yesterday} operator ($\ltl{Y\phi}$, \ie there exists a 
    \emph{previous} state in which $\phi$ holds);
  \item the \emph{weak yesterday} operator ($\ltl{Z\phi}$, \ie either a 
    previous state does not exists or in the previous state $\phi$ holds);
  \item and the \emph{since} operator ($\ltl{\phi_1 S \phi_2}$, \ie there was
    a past state where $\phi_2$ held, and $\phi_1$ has held \emph{since} then).
\end{enumerate*}
We will now briefly recall the syntax and semantics of \LTLP, which encompasses
that of \LTL as well.  Formally, given a set $\Sigma$ of proposition letters,
\LTLP formulas over $\Sigma$ are generated by the following grammar:\fitpar
\begin{align*}
\phi \bydef p
  & \choice \ltl{\lnot\phi}
    \choice \ltl{\phi_1 || \phi_2}
    \choice \ltl{\phi_1 && \phi_2} & \text{propositional connectives} \\
  & \choice \ltl{X\phi_1}
    \choice \ltl{\phi_1 U \phi_2}
    \choice \ltl{\phi_1 R \phi_2}
    \choice \ltl{F \phi_1}
    \choice \ltl{G \phi_1}  & \text{\emph{future} temporal operators} \\
  & \choice \ltl{Y \phi_1}
    \choice \ltl{\phi_1 S \phi_2}
    \choice \ltl{\phi_1 T \phi_2}
    \choice \ltl{O \phi_1}
    \choice \ltl{H \phi_1}
    \choice \ltl{Z \phi_1} & \text{\emph{past} temporal operators}
\end{align*}
where $p\in\Sigma$ and $\phi_1$ and $\phi_2$ are \LTLP formulas. Most of the
temporal operators of the language can be defined in terms of a small number of
basic ones. We refer to~\cite{cimatti2020reactive} for the definition of these
shortcuts.
%
%
We say that an \LTLP formula is \emph{pure past} if and only if all the
temporal operators inside the formula are past operators.  We call \emph{pure
past \LTL}, written as \LTLFP, the fragment of \LTLP containing only pure past
formulas.

Formulas from \LTLP are interpreted over \emph{state sequences}. A \emph{state
sequence} $\sigma=\seq{\sigma_0,\sigma_1,\dots}$ is an \emph{infinite},
linearly ordered sequence of \emph{states}, where each state $\sigma_i$ is
a set of proposition letters, that is $\sigma_i \in 2^\Sigma$ for $i \in \N$.
We will interchangeably use also the term \emph{$\omega$-word} over the
alphabet $2^\Sigma$ for referring to a state sequence. A set of $\omega$-words
is called \emph{$\omega$-language}.
Given two indices $i,j \in \mathbb{Z}$, with $i \le j$, we denote as $\sigma_{[i,j]}$
the interval of $\sigma$ from index $i$ to index $j$, that is
$\seq{\sigma_{i},\dots,\sigma_{j}}$ if $i \ge 0$, or
$\seq{\sigma_{0},\dots,\sigma_{j}}$ otherwise.
With $\sigma_{[i,\infty]}$ we denote the (infinite) suffix of $\sigma$ starting
from $i$.

Given a state sequence $\sigma$, a position $i\ge 0$, and an \LTLP
formula $\phi$, we inductively define the \emph{satisfaction} of $\phi$ by
$\sigma$ at position~$i$, written as $\sigma,i\models\phi$, as follows:
\begin{conditions}
\item $\sigma,i \models p$                 & $p\in\state_i$;
\item $\sigma,i \models \ltl{\lnot\phi}$       & $\sigma,i \not\models \phi$;
\item $\sigma,i \models \ltl{\phi_1 || \phi_2}$  &
        $\sigma,i \models \phi_1$ or $\sigma,i \models \phi_2$;
\item $\sigma,i \models \ltl{X\phi}$     & $\sigma,i+1\models \phi$;
\item $\sigma,i \models \ltl{Y\phi}$    &
        $i > 0$ and $\sigma,i-1\models \phi$;
\item $\sigma,i \models \ltl{\phi_1 U \phi_2}$  &
        there exists $j\ge i$ such that $\sigma,j\models\phi_2$,\newline
        and $\sigma,k\models\phi_1$ for all $k$, with $i \le k < j$;
\item $\sigma,i \models \ltl{\phi_1 S \phi_2}$    &
        there exists $j\le i$ such that $\sigma,j\models\phi_2$,\newline
        and $\sigma,k\models\phi_1$ for all $k$, with $j < k \le i$;
\end{conditions}
We say that $\sigma$ \emph{satisfies} $\phi$, written as $\sigma\models\phi$,
if it satisfies the formula at the first state, \ie if $\sigma,0\models\phi$:
in this case, we call $\sigma$ a \emph{model} of $\phi$. 
We say that two formulas $\phi$ and $\psi$ are \emph{equivalent}
($\phi\equiv\psi$) if and only if they are satisfied by the same set of state
sequences.

If $\phi$ is a full \LTLP formula, then we define the \emph{language} of
$\phi$, written $\lang(\phi)$, as $\lang(\phi) = \{ \sigma \in
(2^\Sigma)^\omega \suchthat \sigma \models \phi \}$. 
%
%
If, instead, $\phi$ contains only past operators, we change the definition of
\emph{language} as follows: for all $\phi \in \LTLFP$, we define the language
over \emph{finite words} of $\phi$ as $\langfin(\phi) \coloneqq \set{\sigma \in
(2^\Sigma)^* \suchthat \sigma = \seq{\sigma_0,\dots,\sigma_n} \land \sigma,n
\models \phi}$.

\paragraph{Notation}
From now on, given a linear temporal logic $\mathbb{L}$, with some
abuse of notation, we will denote with
$\mathbb{L}$ also the set of formulas that \emph{syntactically} belong to
$\mathbb{L}$.  Conversely, we denote with $\sem{\mathbb{L}}$ the set of all and
only those languages $\lang$ of infinite words for which there exists a formula
$\phi \in \mathbb{L}$ (\ie $\phi$ syntactically belongs to $\mathbb{L}$) such
that $\lang = \lang(\phi)$.
For the \LTLFP logic, we write $\semfin{\LTLFP}$ for denoting the set of
languages $\lang$ over \emph{finite} words such that $\lang = \langfin(\phi)$
for some $\phi \in \LTLFP$.

It is known that past modalities do \emph{not} add expressive power to \LTL
\cite{lichtenstein1985glory,gabbay1980temporal,markey2003temporal}, therefore
writing $\sem{\LTL}$ is the same as writing $\sem{\LTLP}$.

\subsection{$\omega$-regular expressions and (co-)Safety classes}

We denote as $\REG$ the set of \emph{regular languages} of finite words
\cite{hopcroft2001introduction}.  
An \emph{$\omega$-regular language} is a set of $\omega$-words
recognized by an $\omega$-regular expression, that is, an expression of the
form $\bigcup_{i=1}^{n} U_i \cdot (V_i)^\omega$, where $n \in \N$ and $U_i,V_i
\in \REG$ for $i=1,\dots,n$.  With \omegaREG, we denote the set of the
$\omega$-regular languages.  One of the seminal results in automata theory is
the correspondence between $\omega$-regular languages and B\"uchi
automata~\cite{buchi1960weak,buchi1960decision}.
An important class of $\omega$-regular languages comprises those languages that
express the fact that something ``bad'' (like for instance a deadlock, or
a simultaneous access into a critical section by two different processes) never
happens. For this reason, they are called \emph{safety languages} (or
\emph{safety properties}).
\begin{definition}[Safety language \cite{kupferman2001model}]
\label{def:safelang}
  Let $\lang \subseteq \Sigma^\omega$ be an $\omega$-regular language. We say
  that $\lang$ is a \emph{safety language} if and only if for all the words
  $\sigma \in \Sigma^\omega$ it holds that, if $\sigma \not \in \lang$, then
  $\exists i \in \N \suchdot \forall \sigma^\prime \in \Sigma^\omega \suchdot
  \sigma_{[0,i]}\cdot\sigma^\prime \not \in \lang$.  The class of safety
  $\omega$-regular languages is denoted as \SAFETY.
\end{definition}
Given some temporal logic $\mathbb{L}$, we say that $\mathbb{L}$ is
a \emph{safety fragment} of $\LTL$ iff
$\phi \in \mathbb{L}$ implies that $\phi \in \LTL$, and $\lang(\phi)$ is
a safety language (\cref{def:safelang}), for all formulas $\phi$.
The class of the $\omega$-regular \emph{co-safety} languages, that we call
\coSAFETY, is defined as the dual of \SAFETY, that is the set of languages
$\lang$ such that $\lang \in \coSAFETY$ iff $\bar{\lang} \in \SAFETY$, where
$\bar{\lang}$ is the complement language of $\lang$.

The \safetyltl logic \cite{sistla1994safety,ZhuTLPV17,ChangMP92} is
defined as the set of \LTL formulas such that, when in negated normal form, do
\emph{not} contain existential temporal operators (\ie $\ltl{U}$ and
$\ltl{F}$). \safetyltl is a \emph{safety fragment} of \LTL
\cite{sistla1994safety}.

We give an alternative and equivalent definition of the \SAFETY class of
\cref{def:safelang}, that will be useful in the following sections: $\SAFETY
\coloneqq \set{\lang \subseteq \Sigma^\omega \suchthat \bar{\lang} = K \cdot
\Sigma^\omega \land K \in \REG}$.

We define the class \SAFETYsf (\coSAFETYsf) as the set obtained from \SAFETY
(resp. \coSAFETY) by restricting $K$ to be a \emph{star-free} expression, that
is, a regular expression devoid of the Kleene star
\cite{mcnaughton1971counter}.  
In particular, $\coSAFETYsf \coloneqq \set{\lang \subseteq \Sigma^\omega
\suchthat \lang = K \cdot \Sigma^\omega \land K \in \SF}$, where $\SF \subseteq
\REG$ is the set of star-free regular expressions.
With \omegaSF we denote the set of star-free $\omega$-regular expressions.
We now state some equivalence results that will be helpful later.  Star-free
expressions (\SF) and pure-past \LTL (\LTLFP) have the same expressive power.
The same holds for the \omegaSF class and \LTL.
\begin{proposition}[Thomas \cite{thomas1988safety}, Lichtenstein \etal \cite{lichtenstein1985glory}]
\label{prop:ltlfpsf}
$\semfin{\LTLFP} = \SF$ and $\sem{\LTL} = \omegaSF$.
\end{proposition}
Finally, we will use the following normal-form theorem, stated in
\cite{ChangMP92}, that proves that any \LTL-definable safety (resp. co-safety)
language can be expressed by a formula of the form $\ltl{G \alpha}$ (resp.
$\ltl{F \alpha}$), and \viceversa. An independent proof of this theorem can be
derived also from the results by Thomas in \cite{thomas1988safety}.
\begin{theorem}[Chang \etal \cite{ChangMP92}]
\label{th:normalformgalpha}
$\sem{\LTL} \cap \SAFETY = \sem{\ltl{G\alpha}}$ and $\sem{\LTL} \cap \coSAFETY
  = \sem{\ltl{F\alpha}}$.
\end{theorem}

\cref{fig:bigpic} summarizes the expressive power of the various fragments and
logics, included \LTLEBRP and \LTLEBR (that are the subject of this paper).

%
%
\begin{figure}[t!]
\centering
\begin{adjustbox}{minipage=\linewidth,scale=1.0}
\begin{tikzpicture}

    \draw[thick] (0.0,0.0) rectangle (12.0,14.0);
    \draw[thick] (1.0,1.0) rectangle (8.0,7.3);
    \draw[thick] (5.0,6.7) rectangle (11.0,13.0);
    \draw[line width=1.8pt] (2.0,3.0) rectangle (10.0,11.0);
    \draw[thick] (4.7,3.2) rectangle (7.7,4.9);
    
    \node at (7.3,13.5) {\vbox
    {\begin{itemize}
      \item[-] \omegaREG
    \end{itemize}}
    };
    \node at (8.15,1.4) {\vbox
    {\begin{itemize}
      \item[-] \SAFETY
    \end{itemize}}
    };
    \node at (9.15,4.3) {\vbox
    {\begin{itemize}
      \item[-] $\sem{\LTLEBRP}$
      \item[-] $\sem{\safetyltl}$
      \item[-] $\sem{\ltl{G\alpha}}$
      \item[-] \SAFETYsf
    \end{itemize}}
    };
    \node at (15.7,12.6) {\vbox
    {\begin{itemize}
      \item[-] \coSAFETY
    \end{itemize}}
    };
    \node at (14.4,10.3) {\vbox
    {\begin{itemize}
      \item[-] \coSAFETYsf
      \item[-] $\sem{\ltl{F\alpha}}$
    \end{itemize}}
    };
    \node at (9.1,9.9) {\vbox
    {\begin{itemize}
      \item[-] \omegaSF
      \item[-] $\sem{\LTL}$
      \item[-] $\sem{\LTLP}$
    \end{itemize}}
    };
    \node at (11.75,4.5) {\vbox
    {\begin{itemize}
      \item[-] \small $\sem{\LTLEBR}$
    \end{itemize}}
    };

\end{tikzpicture}
  \caption{Comparison of expressiveness between the various formalisms. For
  ease of exposition, we highlighted the rectangle corresponding to \LTL with
  thick borders.}
  \label{fig:bigpic}
\end{adjustbox}
\end{figure}
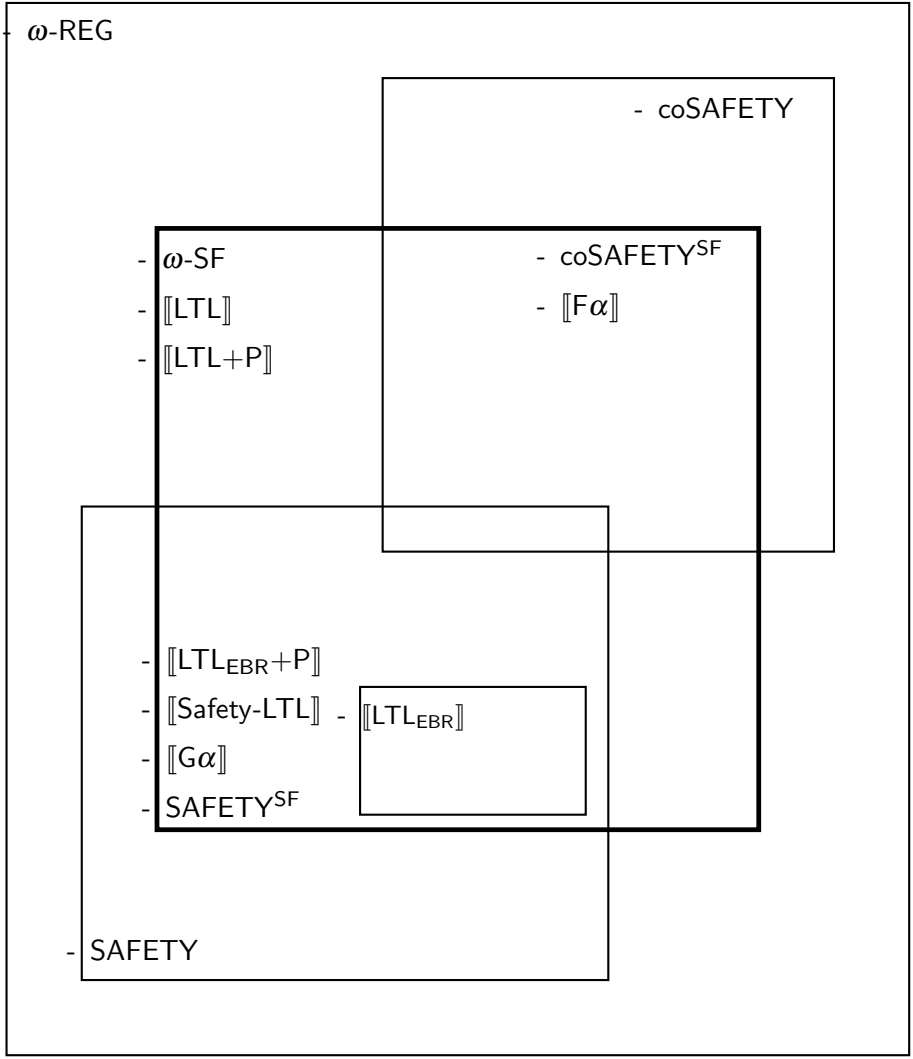

\subsection{Extended Bounded Response \LTL}
\label{sub:ltlebr}

\emph{Extended Bounded Response \LTL with Past} (\LTLEBRP, for short) is
a fragment of \LTLP, 
recently introduced in the context of reactive
synthesis~\cite{cimatti2020reactive}. Here below, we recall its syntax.
\begin{definition}[The logic \LTLEBRP \cite{cimatti2020reactive}]
\label{def:ltlebr}
  Let $a,b \in \N$. An \LTLEBRP formula $\chi$ is inductively defined as
  follows:
  \begin{align*}
    \eta \bydef p 
       &\choice \ltl{\lnot \eta}
       \choice \ltl{\eta_1 || \eta_2}
       \choice \ltl{Y \eta} 
       \choice \ltl{\eta_1 S \eta_2} &
          \text{Pure Past Layer} \\
    \psi \bydef \eta 
       &\choice \ltl{\lnot \psi}
       \choice \ltl{\psi_1 || \psi_2}
       \choice \ltl{X \psi} 
       \choice \ltl{\psi_1 U^{[a,b]}\psi_2} &
          \text{Bounded Future Layer} \\
    \phi \bydef \psi
       &\choice \ltl{\phi_1 \land \phi_2}
       \choice \ltl{X \phi}
       \choice \ltl{G \phi} 
       \choice \ltl{\psi R \phi} &
          \text{Future Layer} \\
    \chi \bydef \phi 
      &\choice \ltl{\chi_1 \lor \chi_2}
      \choice \ltl{\chi_1 \land \chi_2} & 
          \text{Boolean Layer}
  \end{align*}
\end{definition}
We define the \emph{bounded until} operator $\ltl{\psi_1 U^{[a,b]}\psi_2}$ as
a shortcut for the \LTL formula \\ $\bigvee_{i=a}^{b} ( \ltl{X_{1}} \dots
\ltl{X_{i}}(\psi_2) \land \bigwedge_{j=0}^{i-1} \ltl{X_{1}} \dots \ltl{X_{j}}
(\psi_1))$.
This means that \LTLEBRP features really only universal temporal modalities
(\ie $\ltl{X}$, $\ltl{G}$, and $\ltl{R}$), and thus it is a syntactical
fragment of \LTLP and also a \emph{safety} fragment (see Theorem 3.1
in~\cite{sistla1994safety}).
We define \LTLEBR as the fragment of \LTLEBRP devoid of the full past layer.
The syntax of \LTLEBRP is articulated over layers, that impose some syntactical
restrictions on the formulas that can be generated from the grammar. For
example, \LTLEBRP forces the leftmost argument of any \emph{release} operator
to contain no universal temporal modalities (\ie $\ltl{R}$ and $\ltl{G}$).
Originally, the layered structure was guided by the steps of the algorithm for
the construction of symbolic automata starting from \LTLEBRP-formulas. We refer
the reader to \cite{cimatti2020reactive} for more details.

All formulas in \LTLEBRP can be transformed into a \emph{canonical form}
(defined here below) by maintaining the equivalence. \\

\begin{definition}[Canonical Form of \LTLEBRP \cite{cimatti2020reactive}]
\label{def:ltlebrcanonform}
  The \emph{canonical form} of \LTLEBRP is the set of all and only the formulas
  of the following type:
  \begin{align*}
    \ltl{X}^{i_1} \alpha_{i_1} &\otimes \dots \otimes \ltl{X}^{i_j} \alpha_{i_j} \otimes \\
    \ltl{X}^{i_{j+1}} \ltl{G} \alpha_{i_{j+1}} &\otimes \dots \otimes \ltl{X}^{i_k} \ltl{G} \alpha_{i_k} \otimes \\
    \ltl{X}^{i_{k+1}}(\alpha_{i_{k+1}} \ltl{R} \beta_{i_{k+1}}) &\otimes \dots \otimes \ltl{X}^{i_h}(\alpha_{i_h} \ltl{R} \beta_{i_h})
  \end{align*}
  where each $\alpha_i,\beta_i \in \LTLFP$, $\otimes \in \{\land, \lor\}$, and
  $i,j,k,h \in \N$.
\end{definition}
%


\section{Expressive power of \LTLEBRP}
\label{sec:ltlebr}

In this section, we study the expressiveness of the \LTLEBRP logic. In
particular, we compare the set of languages definable in \LTLEBRP
with the set of safety languages expressible in \LTL, and prove that the two
sets are equal, that is 
%
  $\sem{\LTLEBRP} = \sem{\LTL} \cap \SAFETY$. 
%
Consequently, \LTLEBRP and \safetyltl are expressively equivalent
(\ie $\sem{\LTLEBRP} = \sem{\safetyltl}$).


First we recall the normal-form theorem stated in \cref{th:normalformgalpha},
establishing that $\sem{\LTL} \cap \SAFETY = \sem{\ltl{G\alpha}}$.
Proving that $\sem{\LTLEBRP} = \sem{\LTL} \cap \SAFETY$ is straightforward.
In \cite{sistla1994safety}, Sistla proved that any fragment of \LTLP with only
universal (future) temporal operators (\ie $\ltl{X}$, $\ltl{R}$, and $\ltl{G}$)
defines only safety properties, and thus is a safety fragment of \LTLP.
Since \LTLEBRP-formulas contain only universal (future) temporal operators, it
follows that \LTLEBRP is a safety fragment of \LTLP (this corresponds to the
left-to-right direction).  
For the right-to-left direction it suffices to show that the normal form
$\ltl{G\alpha}$ is syntactically definable in \LTLEBRP (\ie $\ltl{G}\alpha \in
\LTLEBRP$ and thus also $\lang(\ltl{G\alpha}) \in \sem{\LTLEBRP}$, for any
$\alpha \in \LTLFP$).
\begin{theorem}
\label{th:ltlebrexpr}
  $\sem{\LTLEBRP} = \sem{\LTL} \cap \SAFETY$.
\end{theorem}
\begin{proof}
  We first prove that $\sem{\LTLEBRP} \subseteq \sem{\LTL} \cap \SAFETY$.  Let
  $\phi \in \sem{\LTLEBRP}$. By \cref{def:ltlebr}, $\phi \in \LTLP$, and thus,
  since $\sem{\LTL} = \sem{\LTLP}$, it holds that $\lang(\phi) \in \sem{\LTL}$.
  Moreover, since \LTLEBRP contains only universal temporal operators, by
  Theorem 3.1 in \cite{sistla1994safety}, it is a \emph{safety} fragment of
  \LTL, and we have that $\lang(\phi) \in \SAFETY$. Therefore, $\lang(\phi) \in
  \sem{\LTL} \cap \SAFETY$.
  
  We now prove that $\sem{\LTL} \cap \SAFETY \subseteq \sem{\LTLEBRP}$.  Let
  $\phi$ be a formula such that $\lang(\phi) \in \sem{\LTL} \cap \SAFETY$.  By
  \cref{th:normalformgalpha}, $\lang(\phi) \in \sem{\ltl{G\alpha}}$. Now,
  $\ltl{G\alpha}$ (for any $\alpha \in \LTLFP$) is a formula that syntactically
  belongs to $\LTLEBRP$, that is $\ltl{G\alpha} \in \LTLEBRP$, and thus
  $\sem{\ltl{G\alpha}} \subseteq \sem{\LTLEBRP}$. It follows that $\lang(\phi)
  \in \sem{\LTLEBRP}$.
\end{proof}

\subsection{Comparison between \LTLEBRP, $\mathsf{G}\alpha$ and \safetyltl}
\label{sub:logicscomparison}

\paragraph{Comparison with $\mathsf{G}\alpha$}

Previously, we proved that the set of languages definable in \LTLEBRP is
exactly the set of safety languages definable in \LTLP.
In turn, \cref{th:normalformgalpha} shows that these sets correspond to
languages definable by a formula of type $\ltl{G\alpha}$, where $\alpha \in
\LTLFP$.
Despite being equivalent fragments, we think that \LTLEBRP offers
a more natural language for safety properties than the $\ltl{G\alpha}$
fragment.
Consider for example the following property, expressed in natural language:
either $p_3$ holds forever, or there exists two time points $t^\prime \le t$
such that
\begin{enumerate*}[label=(\roman*)]
  \item
    $p_1$ holds in $t$,
  \item
    $p_2$ holds in $t^\prime$, and
  \item
    $p_2$ holds from time point $0$ to $t$.
\end{enumerate*}
The property can be easily formalized in \LTLEBRP by the formula $\ltl{p_1
R (p_2 R p_3)}$. The equivalent formula in the $\ltl{G}\alpha$ fragment is
$\ltl{G(H(p_3) \lor O(p_2 \land O(p_1) \land H(p_3)))}$, which is arguably more
intricate.

\paragraph{Comparison with \safetyltl}

\safetyltl is the fragment of \LTL (thus with only future temporal modalities)
containing all and only the \LTL-formulas that, when in negated normal form, do
\emph{not} contain any \emph{until} or \emph{eventually} operator.
In \cite{sistla1994safety}, Sistla proved that this fragment expresses only
safety properties, that is $\sem{\safetyltl} \subseteq \sem{\LTL} \cap
\SAFETY$.  The converse direction, that is $\sem{\LTL} \cap \SAFETY \subseteq
\sem{\safetyltl}$, is reported in \cite{ChangMP92}.
It immediately follows that \LTLEBRP and \safetyltl are expressively
equivalent, namely $\sem{\LTLEBRP} = \sem{\safetyltl}$.

Differently from \LTLEBRP, \safetyltl does \emph{not} impose any syntactic
restriction on the nesting of the logical operators; as a matter of fact,
$\ltl{G(p_1 \lor G p_2)}$ belongs to the syntax of \safetyltl but not to the
syntax of \LTLEBRP, even though $\ltl{G(p_1 \lor G p_2)} \equiv \ltl{G(\lnot
p_2 \to H p_1)} \in \LTLEBRP$.
The restrictions on the syntax of \LTLEBRP are due to algorithmic aspects:
each layer of the syntax of \LTLEBRP (recall \cref{def:ltlebr}) corresponds to
a step of the algorithm for the symbolic automata construction starting from
\LTLEBRP-formulas. As a matter of fact, in practice, \LTLEBRP has shown to
avoid an exponential blowup in time with respect to known algorithms for
automata contruction for safety specifications \cite{cimatti2020reactive}.
Last but not least, the realizability problem of \LTLEBRP is
\EXPTIME-complete~\cite{cimatti2020reactive}, as opposed to
the realizability of \LTLP, which is
\EXPTIME[2]-complete \cite{pnueli1989synthesis,rosner1992modular}.
Consider now \LTLEBR, that is the fragment of \LTLEBRP devoid of past
operators. 
Since each formula of \LTLEBR syntactically belongs to \safetyltl, it
immediately follows that $\sem{\LTLEBR} \subseteq \sem{\safetyltl}$.
In the next section, we will prove that the converse direction does \emph{not}
hold, that is \LTLEBR is \emph{strictly} less expressive than \LTLEBRP, and
thus less expressive than \safetyltl as well.


\section{\LTLEBR is strictly less expressive than full \LTLEBRP}
\label{sec:futureltlebr}

In the previous sections, we have seen that: 
\begin{equation*}
\sem{\LTLEBRP}
  = \sem{\ltl{G\alpha}} 
  = \sem{\LTL} \cap \SAFETY
  = \sem{\safetyltl}
\end{equation*}
In particular, thanks to the use of the \emph{pure past layer} (recall
\cref{def:ltlebr}), \LTLEBRP can easily capture the whole class of
$\sem{\ltl{G\alpha}}$, and thus the whole class of $\sem{\LTL} \cap \SAFETY$.
However, one may wonder whether the pure past layer is really necessary, or
whether the class $\sem{\ltl{G\alpha}}$ can be expressed in \LTLEBRP without
the use of past operators.

\LTLEBR is defined as the fragment of \LTLEBRP devoid of the pure past layer
(recall \cref{sub:ltlebr}). In this section, we investigate the problem of
establishing whether \LTLEBR has the same expressive power of \LTLEBRP, or
equivalently, whether \LTLEBR can express every language in
$\sem{\safetyltl}$.
We will prove that this is \emph{not} the case, that is
\begin{equation}
\label{eq:futureltlebrvsltlebr}
  \sem{\LTLEBR} \subsetneq \sem{\LTLEBRP}
\end{equation}

This result proves that \emph{past modalities}, although being not important
for the expressiveness of full \LTL (since $\sem{\LTL} = \sem{\LTLP}$
\cite{gabbay1980temporal,lichtenstein1985glory,markey2003temporal}), can play
a crucial role for the expressive power of \emph{fragments} of \LTL, like, for
instance, \LTLEBR.

\subsection{The general idea}

We will prove \cref{eq:futureltlebrvsltlebr} by showing that $\sem{\LTLEBR}
\subsetneq \sem{\safetyltl}$. The result in \cref{eq:futureltlebrvsltlebr}
follows from the fact that $\sem{\safetyltl} = \sem{\LTLEBRP}$.
We will prove that the language of the \safetyltl-formula $\phiG \coloneqq
\ltl{G(p_1 \lor G(p_2))}$ cannot be expressed by any \LTLEBR-formula.
The formula $\phiG$ belongs \emph{syntactically} to \safetyltl, and thus
$\lang(\phiG) \in \sem{\safetyltl}$. We also note that $\phiG$ can be expressed
in \LTLEBRP. In fact, it holds that:
\begin{align}
\label{eq:phiGequiv}
  \ltl{G(p_1 \lor G(p_2))} \equiv \ltl{G(\lnot p_2 \to H(p_1))}
\end{align}
Since $\ltl{G(\lnot p_2 \to H(p_1))} \in \LTLEBRP$, it holds that $\lang(\phiG)
\in \sem{\LTLEBRP}$. It is worth noticing the following points:
\begin{enumerate*}[label=(\roman*)]
  \item
    $\ltl{G(\lnot p_2 \to H(p_1))}$ is of the form $\ltl{G\alpha}$, where
    $\alpha \in \LTLFP$ ($\alpha$ is a pure past formula);
  \item
    the formula $\phiG$ is equivalent to $\ltl{G(p_2) || ((X G p_2) R p_1)}$,
    but the latter formula does not syntactically belong to \LTLEBR, due to the
    restriction that forces the leftmost argument of any \emph{release}
    operator to contain no universal temporal operators (\ie $\ltl{R}$ and
    $\ltl{G}$).
\end{enumerate*}
In fact, in the following, we will prove that $\lang(\phiG) \not \in
\sem{\LTLEBR}$.

The proof of the undefinability of $\phiG$ is based on the fact that each
formula of \LTLEBR cannot constrain an arbitrarily long prefix of a state
sequence, but only a finite prefix whose maximum length depends on the maximum
number of nested \emph{next} operators.

Consider again the formula $\phiG \coloneqq \ltl{G(p_1 \lor G(p_2))}$. The
language $\lang(\phiG)$ is expressed by the $\omega$-regular expression
$(\{p_1\})^\omega + (\{p_1\})^*\cdot(\{p_2\})^\omega$. Written in natural
language, each model of $\phiG$ cannot contain a position in which $\lnot p_2$
holds preceded by a position in which $\lnot p_1$ holds.
\begin{remark}
\label{rem:phiGnaturallang}
  Let $\sigma \subseteq (2^\Sigma)^\omega$ be a state sequence. It holds that:
  \begin{align*}
    \sigma \models \phiG \ \Rightarrow \lnot \exists i,j (j \le i \land \sigma_j \models
    \lnot p_1 \land \sigma_i \models \lnot p_2)
  \end{align*}
\end{remark}

We define $^{i,k}\sigma^j$ as \emph{the} state sequence such that at the time
points $i$ and $k$ it holds $p_1 \land \lnot p_2$, at time point $j$ it holds
$\lnot p_1 \land p_2$, and for all the other time points $p_1 \land p_2$ holds. 
The membership of $^{i,k}\sigma^j$ to $\lang(\phiG)$ depends on the value of
the three indices $i$, $j$ and $k$, as follows.  

\begin{remark}
\label{rem:valuesijkmembership}
  If $i<j$ and $k<j$, then $^{i,k}\sigma^j \models \phiG$.  
  Conversely, if $i \ge j$ or $k \ge j$, then $^{i,k}\sigma^j \not\models
  \phiG$.
\end{remark}

As we will see, given a generic formula $\psi \in \LTLEBR$, one can always
find some values for the indices $i$, $j$ and $k$ such that 
\begin{enumerate*}[label=(\alph*)]
  \item
    $j$ is chosen sufficiently greater than $i$;
  \item
    $k$ is chosen sufficiently greater than $j$;
  \item
    $\psi$ is not able to distinguish the state sequence $^{i,i}\sigma^j$ from
    $^{i,k}\sigma^j$.
\end{enumerate*}
Since, by \cref{rem:valuesijkmembership}, $^{i,i}\sigma^j \in \lang(\phiG)$ but
$^{i,k}\sigma^j \not \in \lang(\phiG)$, this proves the undefinability of
$\phiG$ in \LTLEBR.
The rationale is that the \LTLEBR logic combines bounded future formulas (\ie
formulas obtained by a Boolean combination of propositional atoms and $\ltl{X}$
operators) and universal temporal operators (\ie $\ltl{G}$ and $\ltl{R}$).
This implies the fact that, for a generic model $\sigma$ of an \LTLEBR-formula
$\psi$, \emph{at each time point} $i \ge 0$ of $\sigma$ (this corresponds to
the universal temporal operators) only a \emph{finite and bounded suffix} after
$i$ (this corresponds to the \LTLB-formulas) can be constrained by $\psi$ (this
can be thought of as a sort of bounded memory property of this logic).
Equivalently, this means that each \LTLEBR-formula is \emph{not} able to
constrain any finite but arbitrarly long (unbounded) prefix of a state
sequence, contrary, for instance, to the case of the formula $\ltl{G(\lnot p_2
\to H(p_1))}$ (that is equivalent to $\phiG$, see \cref{eq:phiGequiv}).

\subsection{The Canonical Form}
\label{sub:futureltlebrboundmem}

The limitation of \LTLEBR-formulas mentioned before is more evident in the
\emph{canonical form} for the \LTLEBR logic, that we will define in this part.
We first give some preliminaries definitions.
We define \emph{Bounded Past \LTLFP} (\LTLBP, for short) as the set of all and
only the \LTLEBRP formulas that are a Boolean combination of propositional
atoms and \emph{yesterday} operators ($\ltl{Y}$).
We use the shortcut $\ltl{\psi_1 S^{[a,b]} \psi_2}$ for denoting the formula
$\bigvee_{i=a}^{b} ( \ltl{Y_{1}} \dots \ltl{Y_{i}}(\psi_2) \land
\bigwedge_{j=0}^{i-1} \ltl{Y_{1}} \dots \ltl{Y_{j}} (\psi_1))$.
Given a formula $\alpha \in \LTLBP$, we define its \emph{temporal depth},
denoted as $D(\alpha)$, as follows:
\begin{multicols}{2}
\begin{itemize}
  \item $D(p) = 0$, for all $p \in \Sigma$
  \item $D(\lnot\alpha_1) = D(\alpha_1)$
  \item $D(\alpha_1 \land \alpha_2) = \max\{D(\alpha_1),D(\alpha_2)\}$
  \item $D(\ltl{Y\alpha_1}) = 1 + D(\alpha_1)$
  \item $D(\ltl{\alpha_1 S^{[a,b]} \alpha_2}) = b +\max\{D(\alpha_1),D(\alpha_2)\}$
  \item[]
\end{itemize}
\end{multicols}
For each $\alpha \in \LTLBP$, the language $\langfin(\alpha)$ consists only of
words of length at most $D(\alpha)+1$. 
Recall from \cref{sec:prelim} that, given a infinite state sequence $\sigma
= \seq{\sigma_0,\sigma_1,\dots}$ and some $n \ge 0$, $\sigma_{[n-d,n]}$ is the
interval of $\sigma$ of length \emph{at most} $d$ ending at index $n$.
The crucial property of \LTLBP-formulas, that can be shown with a simple
induction, is that their truth over a state sequence $\sigma$ can be checked by
considering only a finite and \emph{bounded} interval of $\sigma$, whose length
depends on the \emph{temporal depth} of the formula.
\begin{remark}
  For any $\alpha \in \LTLBP$, with temporal depth $d=D(\alpha)$, and for any
  $n \ge 0$, it holds that $\sigma,n \models \alpha$ \emph{if and only if}
  $\sigma_{[n-d,n]} \models \alpha$.
\end{remark}

We give now the \emph{canonical form} for \LTLEBR, and we refer to it as
\canLTLEBR.  The canonical form of \LTLEBR forces any universal unbounded
operator, like \emph{globally} or \emph{release}, to contain only
\LTLBP-formulas.  Formally, we define \canLTLEBR as the canonical form
described in \cref{def:ltlebrcanonform} but such that each
$\alpha_i,\beta_i$ is a \emph{bounded past \LTL} formula.
%
%
By applying the same transformation from \LTLEBRP to its canonical form given in
\cite{cimatti2020reactive}, one obtain the following lemma.
\begin{lemma}
\label{lem:futureltlebrcanonform}
  $\sem{\LTLEBR} = \sem{\canLTLEBR}$.
\end{lemma}
\begin{proof}
  Obviously $\sem{\canLTLEBR} \subseteq \sem{\LTLEBR}$, since each formula
  $\psi$ that belongs to \\ $\canLTLEBR$ can be turned
  into an equivalent one $\psi^\prime \in \LTLEBR$ by expanding
  each bounded past operators into conjunctions/disjunctions of
  \emph{yesterday} operators.

  For proving $\sem{\LTLEBR} \subseteq \sem{\canLTLEBR}$, it is sufficient to
  apply the transformations described in \cite{cimatti2020reactive} for the
  translation of \LTLEBRP into canonical form.
  In particular, since by definition $\psi$ has no past temporal operators, the
  only past operators in $\psi^\prime$ are the ones introduced by the
  \emph{pastification} step described in \cite{cimatti2020reactive}, which are all
  \emph{bounded}, that is either $\ltl{Y}$ or $\ltl{S^{[a,b]}}$.
\end{proof}
The canonical form of \LTLEBR makes it easier to prove
\cref{eq:futureltlebrvsltlebr}. Take for example the formula $\ltl{X X G (p
\lor Y p \lor Y Y p)}$, that belongs to \canLTLEBR. It is clear that, at each
time point, this formula can constrain only the interval consisting of the
current state and its two previous states (in fact its temporal depth is $3$).

\subsection{The main proof}
\label{sub:futureltlebrmainproof}

In this part, we show the undefinability of the formula $\phiG$ in the
\canLTLEBR logic.  The undefinability in \LTLEBR follows from
\cref{lem:futureltlebrcanonform}.

Given three indices $i,j,k \in \N$ such that $i \not = j$ and $k \not = j$, we
formally define the state sequence $^{i,k}\sigma^j
= \seq{{}^{i,k}\sigma^j_0 , \ {}^{i,k}\sigma^j_1,\dots}$ as follows:
\begin{align*}
  ^{i,k}\sigma^j_h =
    \begin{cases}
      \{p_1\}    & \text{if } h \in \{i,k\} \\
      \{p_2\}    & \text{if } h = j \\
      \{p_1,p_2\} & \text{otherwise}
    \end{cases}
\end{align*}

The core of the main theorem is based on the fact that any formula of type
$\ltl{G\alpha}$ or $\ltl{\alpha R \beta}$, where $\alpha$ and $\beta$ are
\emph{bounded past \LTLFP} formulas, is not able to distinguish the state
sequence $^{i,i}\sigma^j$ with $i<j$ (which is a model of $\phiG$) from
$^{i,k}\sigma^j$ with $k>j$ (which is \emph{not} a model of $\phiG$), for
sufficiently large values of $i$, $j$ and $k$. The choice for the values of the
three indices is based on the values of the \emph{temporal depth} of $\alpha$
and $\beta$.
Since the \emph{globally} operator is a special case of the \emph{release}
operator, that is $\ltl{G\alpha \equiv \false R \alpha}$, it suffices to prove
the property for formulas of type $\ltl{\alpha R \beta}$.
%
%
We first prove the two fundamental properties that show that, for any interval
of $^{i,i}\sigma^j$ of length at most $d$ (for any $d \in \N$), we can find the
exact same interval in $^{i,k}\sigma^j$, and \viceversa.
\cref{fig:modelsprooffutureltlebr} shows the idea of this correspondence.
%

%
\begin{figure}[t]
  \setlength{\intextsep}{0mm}
\centering
\begin{adjustbox}{minipage=\linewidth,scale=1.0}
\begin{tikzpicture}

    \draw[thick] (1.0,3.0) -- (11.0,3.0);
    \draw[thick] (1.0,2.9) -- (1.0,3.1);
    \draw[thick,dashed] (11.3,3.0) -- (11.8,3.0);
    \node[left] at (0.7,3.0) {$^{i,i}\sigma^j$};
    \node[below] at (1.10,2.9) {$0$};
    \draw[thick] (3.5,2.9) -- (3.5,3.0);
    \draw[thick] (3.8,2.9) -- (3.8,3.0);
    \node[below] at (3.65,2.9) {$i$};
    \node[below] at (3.65,2.5) {\small $\{p_1\}$};
    \draw[thick,dashdotted] (3.0,3.1) -- (4.05,3.1);
    \draw[thick,dotted] (1.4,3.1) -- (2.45,3.1);
    \draw[thick,dotted] (8.5,3.1) -- (9.55,3.1);
    \draw[thick] (6.0,2.9) -- (6.0,3.0);
    \draw[thick] (6.3,2.9) -- (6.3,3.0);
    \node[below] at (6.15,2.9) {$j$};
    \node[below] at (6.15,2.5) {\small $\{p_2\}$};
    \draw[thick,loosely dashdotdotted] (5.5,3.1) -- (6.55,3.1);

    \draw[thick] (1.0,1.0) -- (11.0,1.0);
    \draw[thick] (1.0,0.9) -- (1.0,1.1);
    \draw[thick,dashed] (11.3,1.0) -- (11.8,1.0);
    \node[left] at (0.7,1.0) {$^{i,k}\sigma^j$};
    \node[below] at (1.10,0.9) {$0$};
    \draw[thick] (3.5,0.9) -- (3.5,1.0);
    \draw[thick] (3.8,0.9) -- (3.8,1.0);
    \node[below] at (3.65,0.9) {$i$};
    \node[below] at (3.65,0.5) {\small $\{p_1\}$};
    \draw[thick,dashdotted] (3.0,1.1) -- (4.05,1.1);
    \draw[thick,dotted] (1.4,1.1) -- (2.45,1.1);
    \draw[thick,dotted] (9.9,1.1) -- (10.95,1.1);
    \draw[thick] (6.0,0.9) -- (6.0,1.0);
    \draw[thick] (6.3,0.9) -- (6.3,1.0);
    \node[below] at (6.15,0.9) {$j$};
    \node[below] at (6.15,0.5) {\small $\{p_2\}$};
    \draw[thick,loosely dashdotdotted] (5.5,1.1) -- (6.55,1.1);
    \draw[thick] (9.0,0.9) -- (9.0,1.0);
    \draw[thick] (9.3,0.9) -- (9.3,1.0);
    \node[below] at (9.15,0.9) {$k$};
    \node[below] at (9.15,0.5) {\small $\{p_1\}$};
    \draw[thick,dashdotted] (8.5,1.1) -- (9.55,1.1);

    \node at (13.5,3.0) {Legend};
    \node at (13.0,2.5) {\small Type 1 };
    \draw[thick,dotted] (13.6,2.5) -- (14.65,2.5);
    \node at (13.0,2.0) {\small Type 2 };
    \draw[thick,dashdotted] (13.6,2.0) -- (14.65,2.0);
    \node at (13.0,1.5) {\small Type 3 };
    \draw[thick,loosely dashdotdotted] (13.6,1.5) -- (14.65,1.5);
    \draw (12.2,1.2) rectangle (14.8,3.3);

\end{tikzpicture}
  \caption{}
  \label{fig:modelsprooffutureltlebr}
\end{adjustbox}
\end{figure}

\begin{lemma}
\label{lem:futureltlebrmapping}
  Let $d \in \N$.
  For all $i \ge d$, for all $j \ge i+d$, and for all $k \ge j+d$, it holds
  that:
  \begin{align*}
    &\mbox{Property 1: }
      \forall n^\prime \ge 0 \suchdot \exists n \ge 0 \suchdot
      {}^{i,k}\sigma^{j}_{[n^\prime-d,n^\prime]}
      = {}^{i,i}\sigma^{j}_{[n-d,n]} \\
    &\mbox{Property 2: }
      \forall n \ge 0 \suchdot \exists n^\prime \ge 0 \suchdot
      {}^{i,i}\sigma^{j}_{[n-d,n]}
      = {}^{i,k}\sigma^{j}_{[n^\prime-d,n^\prime]}
  \end{align*}
\end{lemma}
\begin{proof}
  Take any value for $i$, $j$, and $k$ such that:
  \begin{enumerate*}[label=(\roman*)]
    \item $i \ge d$, 
    \item $j \ge i+d$,
    \item $k \ge j+d$.
  \end{enumerate*}
  Given any interval of length $d$ of the state sequence ${}^{i,i}\sigma^{j}$,
  we show how to find an exact same one in ${}^{i,k}\sigma^{j}$, and viceversa.
  
  The constraints above on the three indices ensure that both the state
  sequences $^{i,i}\sigma^j$ and $^{i,k}\sigma^j$ contain \emph{only three}
  types of intervals of length at most $d$.  Consider $^{i,k}\sigma^j$ (the
  case for $^{i,i}\sigma^j$ is specular).  The three types are the following:
  \begin{itemize}[leftmargin=+.7in]
    \item[Type 1:]
      $(\{p_1,p_2\})^n$ for some $0 \le n \le d$;
    \item[Type 2:]
      $(\{p_1,p_2\})^n \cdot (\{p_1\}) \cdot (\{p_1,p_2\})^{d-n-1}$, for some $0 \le
      n < d$;
    \item[Type 3:]
      $(\{p_1,p_2\})^n \cdot (\{p_2\}) \cdot (\{p_1,p_2\})^{d-n-1}$, for some $0 \le
      n < d$;
  \end{itemize}
  The situation is depicted in \cref{fig:modelsprooffutureltlebr}.
  Given any interval of any of the three types above, we show below how to find
  the very same interval in $^{i,i}\sigma^j$
  (\cref{fig:modelsprooffutureltlebr} tries to show visually this
  correspondence):
  \begin{itemize}
    \item
      each interval of $^{i,k}\sigma^j$ of type $(\{p_1,p_2\})^n$ is equal to
      $^{i,i}\sigma^j_{[0,n]}$;
    \item
      each interval of $^{i,k}\sigma^j$ of type $(\{p_1,p_2\})^n \cdot
      (\{p_1\}) \cdot (\{p_1,p_2\})^{d-n-1}$ is equal to
      $^{i,i}\sigma^j_{[i-n,i+d-n-1]}$.
    \item
      each interval of $^{i,k}\sigma^j$ of type $(\{p_1,p_2\})^n \cdot
      (\{p_2\}) \cdot (\{p_1,p_2\})^{d-n-1}$ is equal to
      $^{i,i}\sigma^j_{[j-n,j+d-n-1]}$;
  \end{itemize}
  This proves \emph{Property 1}.
  
  Similarly, the correspondence between intervals of $^{i,i}\sigma^j$ and
  intervals of $^{i,k}\sigma^j$ is the following: 
  \begin{itemize}
    \item
      each interval of $^{i,i}\sigma^j$ of type $(\{p_1,p_2\})^n$ is equal to
      $^{i,k}\sigma^j_{[0,n]}$;
    \item
      each interval of $^{i,i}\sigma^j$ of type $(\{p_1,p_2\})^n \cdot
      (\{p_1\}) \cdot (\{p_1,p_2\})^{d-n-1}$ is equal to
      $^{i,k}\sigma^j_{[i-n,i+d-n-1]}$.
    \item
      each interval of $^{i,i}\sigma^j$ of type $(\{p_1,p_2\})^n \cdot
      (\{p_2\}) \cdot (\{p_1,p_2\})^{d-n-1}$ is equal to
      $^{i,k}\sigma^j_{[j-n,j+d-n-1]}$;
  \end{itemize}
  This proves \emph{Property 2}.
\end{proof}

We can now prove that the state sequences ${}^{i,i}\sigma^{j}$ and
${}^{i,k}\sigma^{j}$ are indistinguishable for each formula of type
$\ltl{\alpha R \beta}$ (and, consequently, of type $\ltl{G \alpha}$), with
$\alpha,\beta \in \LTLBP$.
\begin{lemma}
\label{lem:ltlbpunrec}
  Let $\alpha,\beta \in \LTLBP$, and let $d = \max\{D(\alpha),D(\beta)\}$ be
  the maximum between the temporal depths of $\alpha$ and $\beta$. It holds
  that
  %
    $\ {}^{i,i}\sigma^j \models \ltl{\alpha R \beta}$  
   iff  
    $\ {}^{i,k}\sigma^j \models \ltl{\alpha R \beta}$
  %
  , for all $i \ge d$, for all $j \ge i+d$, and for all $k \ge j+d$.
\end{lemma}
\begin{proof}
  Take any value for $i$, $j$, and $k$ such that:
  \begin{enumerate*}[label=(\roman*)]
    \item $i \ge d$, 
    \item $j \ge i+d$,
    \item $k \ge j+d$.
  \end{enumerate*}
  %

  We first prove the left-to-right direction. Suppose that ${}^{i,i}\sigma^{j}
  \models \ltl{\alpha R \beta}$. We divide in cases:
  \begin{enumerate}
    \item\label{item:ltlbpunrec1}
      Suppose that ${}^{i,i}\sigma^{j}, n \models \beta$ for all $n \ge 0$.
      Since $\beta \in \LTLBP$ and $D(\beta) \le d$, it holds that
      ${}^{i,i}\sigma^{j}_{[n-d,n]} \models \beta$, for all $n \ge 0$.
      Suppose by contradiction that there exists some $n^\prime \ge 0$ such
      that ${}^{i,k}\sigma^{j}_{[n^\prime-d,n^\prime]} \models \lnot \beta$.
      By \emph{Property 1} of \cref{lem:futureltlebrmapping}, this means that
      there exists some $n^{\prime\prime} \ge 0$ such that
      ${}^{i,i}\sigma^{j}_{[n^{\prime\prime}-d,n^{\prime\prime}]} \models \lnot
      \beta$.
      But this is a contradiction. Thus, it holds that
      ${}^{i,k}\sigma^{j}_{[n^\prime-d,n^\prime]} \models \beta$ for all
      $n^\prime \ge 0$, that is, for all $n^\prime \ge 0$, and thus
      ${}^{i,k}\sigma^{j} \models \ltl{\alpha R \beta}$.
    \item\label{item:ltlbpunrec2}
      Suppose that $\exists n \ge 0 \suchdot ({}^{i,i}\sigma^{j},n \models
      \alpha \land \forall 0 \le m \le n \suchdot {}^{i,i}\sigma^{j},m \models
      \beta)$. We divide again in cases:
      \begin{enumerate}
          \item\label{item:ltlbpunrec21}
            Suppose that $n<k$. Then ${}^{i,i}\sigma^{j}_{[0,n]}
            = {}^{i,k}\sigma^{j}_{[0,n]}$. Clearly, it holds that
            ${}^{i,k}\sigma^{j},n \models \alpha$ and ${}^{i,k}\sigma^{j},m
            \models \beta$ for all $0 \le m \le n$. Therefore
            ${}^{i,k}\sigma^{j} \models \ltl{\alpha R \beta}$.
          \item\label{item:ltlbpunrec22}
            Suppose that $n \ge k$. In particular, it holds that
            ${}^{i,i}\sigma^{j}_{[n-d,n]} \models \alpha \land \beta$.
            We use a \emph{contraction argument} for proving that in this case
            there exists a smaller index at which the \emph{release} satisfies
            its existential part (\ie the formula $\alpha$).
            Consider the time point $i-1$. It holds that
            ${}^{i,i}\sigma^{j}_{[i-1-d,i-1]} = {}^{i,i}\sigma^{j}_{[n-d,n]}$
            and thus, since ${}^{i,i}\sigma^{j}_{[n-d,n]} \models \alpha \land
            \beta$ and $\alpha,\beta \in \LTLBP$, we have that
            ${}^{i,i}\sigma^{j}_{[i-1-d,i-1]} \models \alpha \land \beta$.
            Moreover, ${}^{i,i}\sigma^{j}_{[0,i-1]}$ is a prefix of
            ${}^{i,i}\sigma^{j}_{[0,n]}$, and thus, given that
            ${}^{i,i}\sigma^{j}_{[p-d,p]} \models \beta$ for all $0 \le p \le
            n$, it holds that ${}^{i,i}\sigma^{j}_{[p-d,p]} \models \beta$ for
            all $0 \le p \le i-1$.
            From this, it follows that ${}^{i,i}\sigma^{j},i-1 \models
            \alpha$ and ${}^{i,i}\sigma^{j},m \models \beta$ for all $0 \le
            m \le i-1$. Since $i-1 < k$, by \cref{item:ltlbpunrec21}, it holds
            that ${}^{i,k}\sigma^{j} \models \ltl{\alpha R \beta}$.
      \end{enumerate}
  \end{enumerate}
  
  We now prove the right-to-left direction. Suppose that ${}^{i,k}\sigma^{j}
  \models \ltl{\alpha R \beta}$. We divide in cases:
  \begin{enumerate}
    \item\label{item:ltlbpunrec3}
      Suppose that ${}^{i,k}\sigma^{j},n \models \beta$.
      This case is specular to \cref{item:ltlbpunrec1}.
    \item\label{item:ltlbpunrec4}
      Suppose that $\exists n \ge 0 \suchdot ({}^{i,k}\sigma^{j},n \models
      \alpha \land \forall 0 \le m \le n \suchdot {}^{i,k}\sigma^{j},m \models
      \beta)$.
      Since $\alpha,\beta \in \LTLBP$ and $D(\alpha),D(\beta) \le d$, it holds
      that $\exists n \ge 0 \suchdot ({}^{i,k}\sigma^{j}_{[n-d,n]} \models
      \alpha \land \forall 0 \le m \le n \suchdot {}^{i,k}\sigma^{j}_{[m-d,m]}
      \models \beta)$. We divide again in cases:
      \begin{enumerate}
        \item\label{item:ltlbpunrec41}
          If $n<k$, then ${}^{i,k}\sigma^{j}_{[0,n]}
          = {}^{i,i}\sigma^{j}_{[0,n]}$ and thus ${}^{i,i}\sigma^{j},n \models
          \alpha$ and ${}^{i,i}\sigma^{j},m \models \beta$ for all $0 \le m \le
          n$, that is ${}^{i,i}\sigma^{j} \models \ltl{\alpha R \beta}$.
        \item\label{item:ltlbpunrec42}
          If $k \le n \le k+d$, then ${}^{i,k}\sigma^{j}_{[n-d,n]}
          = {}^{i,k}\sigma^{j}_{[n-k-i-d,n-k-i]}$ (we used again a contraction
          argument). Since by hypothesis
          ${}^{i,k}\sigma^{j}_{[n-d,n]} \models \alpha$, it holds also that
          ${}^{i,k}\sigma^{j}_{[n-k-i-d,n-k-i]} \models \alpha$.
          Moreover, ${}^{i,k}\sigma^{j}_{[0,n-k-i]}$ is a prefix of
          ${}^{i,k}\sigma^{j}_{[0,n]}$, and thus, since by hypothesis
          ${}^{i,k}\sigma^{j}_{[p-d,p]} \models \beta$ for all $0 \le p \le n$,
          it also holds that ${}^{i,k}\sigma^{j}_{[p-d,p]} \models \beta$ for
          all $0 \le p \le n-k-i$.
          Therefore ${}^{i,k}\sigma^{j}_{[n-k-i-d,n-k-i]} \models \alpha$ and
          ${}^{i,k}\sigma^{j}_{[m-d,m]} \models \beta$ for all $0 \le m \le
          n-k-i$.
          Since $l+n-i < k$, by \cref{item:ltlbpunrec41}, it holds that
          ${}^{i,i}\sigma^{j} \models \ltl{\alpha R \beta}$.
        \item\label{item:ltlbpunrec43}
          Otherwise $n > k+d$. We have that ${}^{i,k}\sigma^{j}_{[n-d,n]}
          = {}^{i,k}\sigma^{j}_{[i-1,i-1-d]}$ (also in this case we used
          a contraction argument). Since by hypothesis
          ${}^{i,k}\sigma^{j}_{[n-d,n]} \models \alpha$, it also hold that
          ${}^{i,k}\sigma^{j}_{[i-1,i-1-d]} \models \alpha$.
          Moreover ${}^{i,k}\sigma^{j}_{[0,i-1]}$ is a prefix of
          ${}^{i,k}\sigma^{j}_{[0,n]}$ and thus, since by hypothesis
          ${}^{i,k}\sigma^{j}_{[p-d,p]} \models \beta$ for all $0 \le p \le n$,
          it also holds that ${}^{i,k}\sigma^{j}_{[p-d,p]} \models \beta$ for
          all $0 \le p \le i-1$.
          Therefore ${}^{i,k}\sigma^{j},i-1 \models \alpha$ and
          ${}^{i,k}\sigma^{j},m \models \beta$ for all $0 \le m \le i-1$.
          Since $i-1<k$, by \cref{item:ltlbpunrec41}, it holds that
          ${}^{i,i}\sigma^{j} \models \ltl{\alpha R \beta}$.
      \end{enumerate}
  \end{enumerate}
  %
\end{proof}

By using \cref{lem:ltlbpunrec} as the proof for the base case, we prove by
induction on the structure of the formula that any formula in \canLTLEBR is not
able to distinguish the state sequences ${}^{i,i}\sigma^j$ and
${}^{i,k}\sigma^j$ for sufficiently large values of $i,j,k$.
In the following, given a formula $\psi \in \canLTLEBR$, we will denote with
$m_\psi$ the maximum number of nested \emph{next} operators in $\psi$, and
with $d_\psi$ the maximum temporal depth between all its \LTLBP-subformulas.

\begin{restatable}{lemma}{canonfutureltlunreac}
\label{lem:canonfutureltlunreac}
  Let $\psi \in \canLTLEBR$. It holds that
  %
    ${}^{i,i}\sigma^j \models \psi$ 
   iff 
    ${}^{i,k}\sigma^j \models \psi$
  %
  , for all $i \ge m_\psi + d_\psi$, for all $j \ge i+d_\psi$, and for all $k
  \ge j+d_\psi$.
\end{restatable}
\begin{proof}
  Take any value for $i$, $j$, and $k$ such that:
  \begin{enumerate*}[label=(\roman*)]
    \item $i \ge m_\psi + d_\psi$, 
    \item $j \ge i+d_\psi$,
    \item $k \ge j+d_\psi$.
  \end{enumerate*}
  We proceed by induction on the structure of the formula $\psi$.

  For the base case, we consider three cases:
  \begin{enumerate*}[label=(\roman*)]
    \item
      formulas in $\LTLBP$, that is such that all its temporal operators refer
      to the past and are bounded;
    \item
      formulas of type $\ltl{G\alpha}$, where $\alpha \in \LTLBP$;
    \item
      formulas of type $\ltl{\alpha R \beta}$, where $\alpha,\beta \in \LTLBP$;
  \end{enumerate*}

  We consider the case of a formula $\alpha \in \LTLBP$, and suppose that
  ${}^{i,i}\sigma^j \models \alpha$.
  By definition of ${}^{i,i}\sigma^j$ and ${}^{i,k}\sigma^j$, it always holds
  that ${}^{i,i}\sigma^{j}_0 = {}^{i,k}\sigma^{j}_0$. Since $\alpha \in \LTLBP$
  refers only to the current state or to the past, it follows that
  ${}^{i,i}\sigma^{j} \models \alpha$ if and only if ${}^{i,k}\sigma^{j}
  \models \alpha$.

  Consider now the case for $\ltl{\alpha R \beta}$, where $\alpha,\beta \in
  \LTLBP$. Since $m_{\ltl{\alpha R \beta}} = 0$ (\ie the are no \emph{next}
  operators in this formula), we can apply \cref{lem:ltlbpunrec}, having that
  ${}^{i,i}\sigma^j \models \ltl{\alpha R \beta}$ if and only if
  ${}^{i,k}\sigma^j \models \ltl{\alpha R \beta}$.  Since $\ltl{G\alpha}
  = \ltl{\false R \alpha}$, this proves also the case for the \emph{globally}
  operator.

  For the inductive step, since by hypothesis $\psi$ belongs to the canonical
  form of \LTLEBR, it suffices to consider only the case for the \emph{next
  operator}, \emph{conjunctions} and \emph{disjunctions}.

  Consider first the case for the \emph{next} operator, and suppose that
  ${}^{i,i}\sigma^j \models \ltl{X\psi^\prime}$. For any indices $k$, $i$ and $j$
  such that $i \ge m_{\ltl{X\psi^\prime}}+d_{\ltl{X\psi^\prime}}$, $j \ge
  i+d_{\ltl{X\psi^\prime}}$ and $k \ge j+d_{\ltl{X\psi^\prime}}$, we want to
  prove that $^{i,k}\sigma^{j} \models \ltl{X\psi^\prime}$.
  By definition of the next operator, it holds that ${}^{i,i}\sigma^j,1 \models
  \psi^\prime$. Now, let $\tau$ be the state sequence obtained from
  $^{i,i}\sigma^j$ by discarding its initial state, that is $\tau \coloneqq
  {}^{i,i}\sigma^j_{[1,\infty)}$. Obviously, $\tau \models \psi^\prime$.
  We observe that $\tau$ is equal to the state sequence $^{i-1,i-1}\sigma^{j-1}$.
  Since the maximum number $m_{\psi^\prime}$ of nested next operators in
  $\psi^\prime$ is $m_{\ltl{X\psi^\prime}}-1$ (while $\alpha_{\psi^\prime}$
  remains the same), we can apply the inductive hypothesis on $\psi^\prime$,
  having that $^{i-1,k-1}\sigma^{j-1} \models \psi^\prime$.
  By definition of $\tau$, it follows that $^{i,k}\sigma^{j} \models
  \ltl{X}\psi^\prime$.

  We consider now the case for conjunctions, and suppose that ${}^{i,i}\sigma^j
  \models \psi_1 \land \psi_2$, for generic indices $k$, $i$ and $j$ such that
  $i \ge m_{\psi_1 \land \psi_2} + d_{\psi_1 \land \psi_2}$, $j \ge
  i + d_{\psi_1 \land \psi_2}$, and $k \ge j + d_{\psi_1 \land \psi_2}$.
  It holds that ${}^{i,i}\sigma^j \models \psi_1$ and ${}^{i,i}\sigma^j \models
  \psi_2$.  Moreover, $m_{\psi_1} \le m_{\psi_1 \land \psi_2}$ and $m_{\psi_2}
  \le m_{\psi_1 \land \psi_2}$. Similarly, $d_{\psi_1} \le d_{\psi_1 \land
  \psi_2}$ and $d_{\psi_2} \le d_{\psi_1 \land \psi_2}$. This means that we can
  apply the inductive hypothesis both on $\psi_1$ and $\psi_2$ on the
  \emph{current} indices $k$, $i$ and $j$.
  By inductive hypothesis, we have that $^{i,k}\sigma^j \models \psi_1$ and
  $^{i,k}\sigma^j \models \psi_2$. It follows that $^{i,k}\sigma^j \models
  \psi_1 \land \psi_2$.
  The case for $\psi_1 \lor \psi_2$ is specular.
  %
\end{proof}

Thanks to \cref{lem:canonfutureltlunreac}, it is simple to prove the
undefinability of $\ltl{G(p_1 \lor G(p_2))}$ in \LTLEBR, that proves that
\LTLEBR is \emph{strictly less expressive} than \safetyltl.

\begin{theorem}
\label{th:futureltlebrvssafetyltl}
  $\sem{\LTLEBR} \subsetneq \sem{\safetyltl}$.
\end{theorem}
\begin{proof}
  Consider the formula $\phiG \coloneqq \ltl{G(p_1 \lor G(p_2))}$.
  We prove that there does \emph{not} exists a formula $\psi \in \LTLEBR$
  such that $\lang(\psi) = \lang(\phiG)$.
  We proceed by contradiction. Suppose that there exists a formula $\psi \in
  \LTLEBR$ such that $\lang(\psi) = \lang(\phiG)$.
  By \cref{lem:futureltlebrcanonform}, there exists a formula $\psi^{\prime}
  \in \canLTLEBR$ such that $\lang(\psi) = \lang(\psi^\prime)$.
  Let $m_{\psi^\prime}$ be the maximum number of \emph{nested next} operators
  in $\psi^\prime$, and let $d_{\psi^\prime}$ be the maximum temporal depth
  between all the \LTLBP-subformulas in $\psi^\prime$.
  Let $k$, $i$ and $j$ be three indices such that:
  \begin{enumerate*}[label=(\roman*)]
    \item $i \ge m_{\psi^\prime} + d_{\psi^\prime}$;
    \item $j \ge i + d_{\psi^\prime}$;
    \item and $k \ge j + d_{\psi^\prime}$.
  \end{enumerate*}
  Consider the two state sequences ${}^{i,i}\sigma^j$ and ${}^{i,k}\sigma^j$.
  By \cref{lem:canonfutureltlunreac}, ${}^{i,i}\sigma^j \in \lang(\psi^\prime)$
  if and only if ${}^{i,k}\sigma^j \in \lang(\psi^\prime)$, that is
  ${}^{i,i}\sigma^j \in \lang(\phiG)$ if and only if ${}^{i,k}\sigma^j \in
  \lang(\phiG)$. 
  Since it holds that ${}^{i,i}\sigma^j \in \lang(\phiG)$ but ${}^{i,k}\sigma^j
  \not \in \lang(\phiG)$, this is clearly a contradiction.
\end{proof}

\begin{corollary}
\label{cor:futureltlebrvsltlebr}
  $\sem{\LTLEBR} \subsetneq \sem{\LTLEBRP}$.
\end{corollary}


\section{Conclusions}
\label{sec:conc}

We considered the logic \LTLEBRP, a recently introduced safety fragment of \LTL
with an efficient realizability problem.  The syntax of \LTLEBRP made it
difficult to exactly characterize its expressive power.  We studied the
expressive power of \LTLEBRP and of its pure future fragment, \LTLEBR, and
compare it with other safety fragments of \LTL.
It turned out that \LTLEBRP is expressively complete with respect to the safety
fragment of \LTL, and, consequently, it is expressively equivalent to
\safetyltl.
We found out that past modalities are crucial for the expressive power of
\LTLEBRP. In fact, \LTLEBR is strictly less expressive than full \LTLEBRP.
This was somehow surprising, since it proves that, despite not being
fundamental for the expressiveness of full \LTL, past modalities are crucial
for fragments of \LTL, like, for instance, \LTLEBRP.

\bibliographystyle{eptcs}
\bibliography{biblio}

\begin{thebibliography}{10}
\providecommand{\bibitemdeclare}[2]{}
\providecommand{\surnamestart}{}
\providecommand{\surnameend}{}
\providecommand{\urlprefix}{Available at }
\providecommand{\url}[1]{\texttt{#1}}
\providecommand{\href}[2]{\texttt{#2}}
\providecommand{\urlalt}[2]{\href{#1}{#2}}
\providecommand{\doi}[1]{doi:\urlalt{http://dx.doi.org/#1}{#1}}
\providecommand{\eprint}[1]{ArXiv:\urlalt{https://arxiv.org/abs/#1}{#1}}
\providecommand{\bibinfo}[2]{#2}

\bibitemdeclare{article}{buchi1960weak}
\bibitem{buchi1960weak}
\bibinfo{author}{J~Richard \surnamestart B{\"u}chi\surnameend}
  (\bibinfo{year}{1960}): \emph{\bibinfo{title}{Weak Second-Order Arithmetic
  and Finite Automata}}.
\newblock {\sl \bibinfo{journal}{Mathematical Logic Quarterly}}
  \bibinfo{volume}{6}(\bibinfo{number}{1-6}), pp. \bibinfo{pages}{66--92},
  \doi{10.1002/malq.19600060105}.

\bibitemdeclare{article}{buchi1960decision}
\bibitem{buchi1960decision}
\bibinfo{author}{J.~Richard \surnamestart Büchi\surnameend}
  (\bibinfo{year}{1960}): \emph{\bibinfo{title}{On a Decision Method in
  Restricted Second Order Arithmetic}}, pp. \bibinfo{pages}{425--435}.
\newblock \doi{10.1007/978-1-4613-8928-6_23}.
\newblock \urlprefix\url{https://doi.org/10.1007%2F978-1-4613-8928-6_23}.

\bibitemdeclare{inproceedings}{ChangMP92}
\bibitem{ChangMP92}
\bibinfo{author}{Edward~Y. \surnamestart Chang\surnameend},
  \bibinfo{author}{Zohar \surnamestart Manna\surnameend} \&
  \bibinfo{author}{Amir \surnamestart Pnueli\surnameend}
  (\bibinfo{year}{1992}): \emph{\bibinfo{title}{Characterization of Temporal
  Property Classes}}.
\newblock In \bibinfo{editor}{Werner \surnamestart Kuich\surnameend}, editor:
  {\sl \bibinfo{booktitle}{Proceedings of the 19th International Colloquium on
  Automata, Languages and Programming}}, {\sl \bibinfo{series}{Lecture Notes in
  Computer Science}} \bibinfo{volume}{623}, \bibinfo{publisher}{Springer}, pp.
  \bibinfo{pages}{474--486}, \doi{10.1007/3-540-55719-9\_97}.

\bibitemdeclare{inproceedings}{cimatti2020reactive}
\bibitem{cimatti2020reactive}
\bibinfo{author}{Alessandro \surnamestart Cimatti\surnameend},
  \bibinfo{author}{Luca \surnamestart Geatti\surnameend},
  \bibinfo{author}{Nicola \surnamestart Gigante\surnameend},
  \bibinfo{author}{Angelo \surnamestart Montanari\surnameend} \&
  \bibinfo{author}{Stefano \surnamestart Tonetta\surnameend}
  (\bibinfo{year}{2020}): \emph{\bibinfo{title}{Reactive Synthesis from
  Extended Bounded Response {LTL} Specifications}}.
\newblock In: {\sl \bibinfo{booktitle}{2020 Formal Methods in Computer Aided
  Design, {FMCAD} 2020, Haifa, Israel, September 21-24, 2020}},
  \bibinfo{publisher}{{IEEE}}, pp. \bibinfo{pages}{83--92},
  \doi{10.34727/2020/isbn.978-3-85448-042-6\_15}.

\bibitemdeclare{book}{demri2016temporal}
\bibitem{demri2016temporal}
\bibinfo{author}{St{\'e}phane \surnamestart Demri\surnameend},
  \bibinfo{author}{Valentin \surnamestart Goranko\surnameend} \&
  \bibinfo{author}{Martin \surnamestart Lange\surnameend}
  (\bibinfo{year}{2016}): \emph{\bibinfo{title}{Temporal logics in computer
  science: finite-state systems}}.
\newblock \bibinfo{volume}{58}, \bibinfo{publisher}{Cambridge University
  Press}, \doi{10.1017/CBO9781139236119}.

\bibitemdeclare{inproceedings}{gabbay1980temporal}
\bibitem{gabbay1980temporal}
\bibinfo{author}{Dov~M. \surnamestart Gabbay\surnameend}, \bibinfo{author}{Amir
  \surnamestart Pnueli\surnameend}, \bibinfo{author}{Saharon \surnamestart
  Shelah\surnameend} \& \bibinfo{author}{Jonathan \surnamestart
  Stavi\surnameend} (\bibinfo{year}{1980}): \emph{\bibinfo{title}{On the
  Temporal Analysis of Fairness}}.
\newblock In \bibinfo{editor}{Paul~W. \surnamestart Abrahams\surnameend},
  \bibinfo{editor}{Richard~J. \surnamestart Lipton\surnameend} \&
  \bibinfo{editor}{Stephen~R. \surnamestart Bourne\surnameend}, editors: {\sl
  \bibinfo{booktitle}{Conference Record of the Seventh Annual {ACM} Symposium
  on Principles of Programming Languages, Las Vegas, Nevada, USA, January
  1980}}, \bibinfo{publisher}{{ACM} Press}, pp. \bibinfo{pages}{163--173},
  \doi{10.1145/567446.567462}.

\bibitemdeclare{article}{hopcroft2001introduction}
\bibitem{hopcroft2001introduction}
\bibinfo{author}{John~E \surnamestart Hopcroft\surnameend},
  \bibinfo{author}{Rajeev \surnamestart Motwani\surnameend} \&
  \bibinfo{author}{Jeffrey~D \surnamestart Ullman\surnameend}
  (\bibinfo{year}{2001}): \emph{\bibinfo{title}{Introduction to automata
  theory, languages, and computation}}.
\newblock {\sl \bibinfo{journal}{Acm Sigact News}}
  \bibinfo{volume}{32}(\bibinfo{number}{1}), pp. \bibinfo{pages}{60--65},
  \doi{10.1145/568438.568455}.

\bibitemdeclare{article}{kamp1968tense}
\bibitem{kamp1968tense}
\bibinfo{author}{Johan Anthony~Wilem \surnamestart Kamp\surnameend}
  (\bibinfo{year}{1968}): \emph{\bibinfo{title}{Tense logic and the theory of
  linear order}}.

\bibitemdeclare{article}{kupferman2001model}
\bibitem{kupferman2001model}
\bibinfo{author}{Orna \surnamestart Kupferman\surnameend} \&
  \bibinfo{author}{Moshe~Y \surnamestart Vardi\surnameend}
  (\bibinfo{year}{2001}): \emph{\bibinfo{title}{Model checking of safety
  properties}}.
\newblock {\sl \bibinfo{journal}{Formal Methods in System Design}}
  \bibinfo{volume}{19}(\bibinfo{number}{3}), pp. \bibinfo{pages}{291--314},
  \doi{10.1023/A:1011254632723}.

\bibitemdeclare{inproceedings}{lichtenstein1985glory}
\bibitem{lichtenstein1985glory}
\bibinfo{author}{Orna \surnamestart Lichtenstein\surnameend},
  \bibinfo{author}{Amir \surnamestart Pnueli\surnameend} \&
  \bibinfo{author}{Lenore \surnamestart Zuck\surnameend}
  (\bibinfo{year}{1985}): \emph{\bibinfo{title}{The glory of the past}}.
\newblock In: {\sl \bibinfo{booktitle}{Workshop on Logic of Programs}},
  \bibinfo{organization}{Springer}, pp. \bibinfo{pages}{196--218},
  \doi{10.1007/3-540-15648-8_16}.

\bibitemdeclare{article}{markey2003temporal}
\bibitem{markey2003temporal}
\bibinfo{author}{Nicolas \surnamestart Markey\surnameend}
  (\bibinfo{year}{2003}): \emph{\bibinfo{title}{Temporal logic with past is
  exponentially more succinct}}.

\bibitemdeclare{book}{mcnaughton1971counter}
\bibitem{mcnaughton1971counter}
\bibinfo{author}{Robert \surnamestart McNaughton\surnameend} \&
  \bibinfo{author}{Seymour~A \surnamestart Papert\surnameend}
  (\bibinfo{year}{1971}): \emph{\bibinfo{title}{Counter-Free Automata (MIT
  research monograph no. 65)}}.
\newblock \bibinfo{publisher}{The MIT Press}.

\bibitemdeclare{inproceedings}{pnueli1977temporal}
\bibitem{pnueli1977temporal}
\bibinfo{author}{Amir \surnamestart Pnueli\surnameend} (\bibinfo{year}{1977}):
  \emph{\bibinfo{title}{The temporal logic of programs}}.
\newblock In: {\sl \bibinfo{booktitle}{18th Annual Symposium on Foundations of
  Computer Science (sfcs 1977)}}, \bibinfo{organization}{IEEE}, pp.
  \bibinfo{pages}{46--57}, \doi{10.1109/SFCS.1977.32}.

\bibitemdeclare{inproceedings}{pnueli1989synthesis}
\bibitem{pnueli1989synthesis}
\bibinfo{author}{Amir \surnamestart Pnueli\surnameend} \& \bibinfo{author}{Roni
  \surnamestart Rosner\surnameend} (\bibinfo{year}{1989}):
  \emph{\bibinfo{title}{On the synthesis of an asynchronous reactive module}}.
\newblock In: {\sl \bibinfo{booktitle}{International Colloquium on Automata,
  Languages, and Programming (ICALP)}}, \bibinfo{organization}{Springer}, pp.
  \bibinfo{pages}{652--671}, \doi{10.1016/0022-0000(86)90026-7}.

\bibitemdeclare{phdthesis}{rosner1992modular}
\bibitem{rosner1992modular}
\bibinfo{author}{Roni \surnamestart Rosner\surnameend} (\bibinfo{year}{1992}):
  \emph{\bibinfo{title}{Modular synthesis of reactive systems}}.
\newblock Ph.D. thesis, \bibinfo{school}{PhD thesis, Weizmann Institute of
  Science}.

\bibitemdeclare{article}{sistla1994safety}
\bibitem{sistla1994safety}
\bibinfo{author}{A~Prasad \surnamestart Sistla\surnameend}
  (\bibinfo{year}{1994}): \emph{\bibinfo{title}{Safety, liveness and fairness
  in temporal logic}}.
\newblock {\sl \bibinfo{journal}{Formal Aspects of Computing}}
  \bibinfo{volume}{6}(\bibinfo{number}{5}), pp. \bibinfo{pages}{495--511},
  \doi{10.1007/BF01211865}.

\bibitemdeclare{article}{thomas1988safety}
\bibitem{thomas1988safety}
\bibinfo{author}{Wolfgang \surnamestart Thomas\surnameend}
  (\bibinfo{year}{1988}): \emph{\bibinfo{title}{Safety-and liveness-properties
  in propositional temporal logic: characterizations and decidability}}.
\newblock {\sl \bibinfo{journal}{Banach Center Publications}}
  \bibinfo{volume}{1}(\bibinfo{number}{21}), pp. \bibinfo{pages}{403--417},
  \doi{10.4064/-21-1-403-417}.

\bibitemdeclare{inproceedings}{ZhuTLPV17}
\bibitem{ZhuTLPV17}
\bibinfo{author}{Shufang \surnamestart Zhu\surnameend},
  \bibinfo{author}{Lucas~M. \surnamestart Tabajara\surnameend},
  \bibinfo{author}{Jianwen \surnamestart Li\surnameend},
  \bibinfo{author}{Geguang \surnamestart Pu\surnameend} \&
  \bibinfo{author}{Moshe~Y. \surnamestart Vardi\surnameend}
  (\bibinfo{year}{2017}): \emph{\bibinfo{title}{{A Symbolic Approach to Safety
  LTL Synthesis}}}.
\newblock In \bibinfo{editor}{Ofer \surnamestart Strichman\surnameend} \&
  \bibinfo{editor}{Rachel \surnamestart Tzoref{-}Brill\surnameend}, editors:
  {\sl \bibinfo{booktitle}{Proceedings of the 13th International Haifa
  Verification Conference}}, {\sl \bibinfo{series}{Lecture Notes in Computer
  Science}} \bibinfo{volume}{10629}, \bibinfo{publisher}{Springer}, pp.
  \bibinfo{pages}{147--162}, \doi{10.1007/978-3-319-70389-3\_10}.

\end{thebibliography}


\end{document}